\documentclass[aps,epsfig,floatfix,a4paper]{article}

\pdfoutput=1

\usepackage{titlesec}
\usepackage[latin1]{inputenc}
\usepackage{enumerate}
\usepackage{graphicx}
\usepackage{subfigure}
\usepackage{rotating}
\usepackage{amsmath,amsthm,amsmath,amsfonts,amssymb,color,graphicx}
\usepackage[usenames,dvipsnames]{pstricks}
\usepackage{epsfig}
\usepackage{pst-grad}
\usepackage{pst-plot}
\usepackage[final]{pdfpages}

\newtheorem{theorem}{Theorem}

\newtheorem{lemma}[theorem]{Lemma}

\newtheorem*{remark}{Remark}

\newtheorem*{example}{Example}

\def\RR{\mathbb{R}}
\def\CC{\mathbb{C}}

\def\bff{\mathbf{f}}
\def\bfh{\mathbf{h}}
\def\bfu{\mathbf{u}}
\def\bfv{\mathbf{v}}
\def\bfw{\mathbf{w}}
\def\mcd{\mathcal{S}_d}

\begin{document}

\title{\bf Some non-standard ways to generate SIC-POVMs in dimensions 2 and 3}

\author{Gary McConnell\\\it
Controlled Quantum Dynamics Theory Group\\\it
Imperial College London\\\rm
\texttt{g.mcconnell@imperial.ac.uk}}



\date{\today}


\maketitle


\bf The notion of Symmetric Informationally Complete Positive Operator-Valued Measures (SIC-POVMs) arose in physics as a kind of optimal measurement basis for quantum systems~\cite{zauner, renes}. However the question of the existence of such systems is identical to the question of the existence of a maximal set of~\emph{complex equiangular lines}. That is to say, given a complex Hilbert space of dimension~$\mathbf d$, what is the maximal number of (complex) lines one can find which all make a common (real) angle with one another, in the sense that the inner products between unit vectors spanning those lines all have a common absolute value? A maximal set would consist of~$\mathbf d^2$ lines all with a common angle, the absolute value of whose cosine is equal to~$\mathbf{\frac{1}{\sqrt{d+1}}}$. The same question has also been posed in the real case and some partial answers are known: see~\cite[A002853]{sloan} for the known results; and for some of the theory see~\cite[chapter 11]{godsil}. But at the time of writing no unifying theoretical result has been found in the real or the complex case: some sporadic low-dimensional numerical constructions have been converted into algebraic solutions but beyond this very little is known. It is conjectured~\cite{renes, zauner} that such maximal structures always arise as orbits of certain fiducial vectors under the action of the Weyl (or generalised Pauli) group. In this paper we point out some new construction methods in the lowest dimensions ($\mathbf{ d=2}$ and $\mathbf {d=3}$). We should mention that the SIC-POVMs so constructed are all unitarily equivalent to previously known SIC-POVMs. \rm

\section*{SIC-POVMs and Complex Equiangular Lines}

Let~$d>1$ be a positive integer and let~$\CC^d$ denote complex Hilbert space of dimension~$d$ equipped with the usual Hermitian positive-definite inner product, denoted by~$\langle\ ,\ \rangle$. A \emph{complex line} is a 1-dimensional complex subspace of~$\CC^d$. We shall view such a line as being spanned by a unit vector~$\bfu$ which is unique up to a~\emph{phase} (an element of the complex unit circle). As in the real case (where the phase ambiguity however only extends to~$\pm 1$) we may ask about the relative \emph{angle} between two such complex lines. Although the definition of such angles is open to several interpretations~\cite{scharnhorst} we shall adopt the usual convention here and define the angle~$\alpha_{\bfu,\bfv}$ between two lines spanned by unit vectors $\bfu,\bfv$ to be the inverse cosine of the absolute value of their Hermitian inner product $\langle\bfu,\bfv\rangle$, viz:
$$
\alpha_{\bfu,\bfv} = \arccos\left(\vert\langle\bfu,\bfv\rangle\vert\right).
$$
Notice that this definition is unchanged if we multiply $\bfu$ or $\bfv$ or both by (possibly distinct) phases. We follow Scharnhorst~\cite{scharnhorst} in referring to $\alpha_{\bfu,\bfv}$ as the \emph{Hermitian angle} between the vectors~$\bfu$ and~$\bfv$. In~\cite{renes} it is shown that the generating set $\mcd$ of unit vectors for a complete (maximal) set of equiangular lines in $\CC^d$ will necessarily have cardinality $d^2$ and each pair of distinct vectors $\bfu,\bfv$ will satisfy
\begin{equation}\label{SICcond}
\vert\langle\bfu,\bfv\rangle\vert = \frac{1}{\sqrt{(d+1)}}.
\end{equation}

We shall speak about SIC-POVMs and complete sets of equiangular lines as though they were the same object: the translation from one perspective to another may be found in~\cite{renes}. Also where it will cause no confusion we shall not distinguish between row and column vectors, to avoid cluttering up the exposition with transpose symbols. To illustrate the basic idea we shall look at the simplest non-trivial real Euclidean example.

\begin{example}
Let $d=2$ and consider $\RR^2$ equipped with the usual inner (dot) product. Then the three unit vectors $(1,0)$,  $(\frac{1}{2},\frac{\sqrt{3}}{2})$ and $(-\frac{1}{2},\frac{\sqrt{3}}{2})$ span three one-dimensional subspaces which constitute a (maximal) set of 3 equiangular lines in~$\RR^2$, with the mutual angle between them being $\arccos(\frac{1}{2})=\frac{\pi}{3}$.
\end{example}

In~\cite{renes} we find the first systematic numerical search for SIC-POVMs in low dimensions, with the smaller dimensional examples being converted into complete algebraic solutions. This was followed by~\cite{appleby05},~\cite{scottgrassl} and~\cite{appleby12} (the literature is in fact much broader: for a much more extensive set of references see~\cite{appleby12}). The framework in which all of this previous work has been completed is that of the action of the standard $d$-dimensional (Heisenberg-)Weyl Group $W_d$ upon a single \emph{fiducial vector} $\bff_d$: the orbit (modulo phases) of $\bff_d$ under the action of $W_d$ is then the entire SIC-POVM. Hence the focus has been upon finding such fiducial vectors $\bff_d$ since the basis for the expression of the $X_d$ and $Z_d$ matrices which generate $W_d$ is assumed fixed, hence the numerics can focus on just one vector in each dimension.

The focus of this paper is somewhat different: we explore some other ways in which such structures can arise in dimensions 2 and 3. The original idea behind these constructions was to try to find a way of generating all of the elements of a SIC-POVM from a single matrix, by somehow creating a (not necessarily unitary) matrix which takes a simple vector like $\bfv_0=(1,0,\ldots,0)$ and then successively `twists' it to new vectors which have the appropriate angle to all of the previous ones. As we shall see below, this was possible for $d=2$ but is too ambitious for higher dimensions, even for $d=3$. So instead we built SIC-POVMs starting with $\bfv_0$ and building up in a sequence via simple geometric steps, based on the single-matrix dimension~2 example, which give the appropriate angles as we go along. Once again, this works in dimensions~2 and~3 but so far we have not been able to generalise the method to higher dimensions. However it points to a possible new heuristic for achieving such constructions in the general case.

Thanks to Marcus Appleby for valuable discussions and for his comments on an earlier draft of this paper. I would also like to thank Terry Rudolph for many helpful ideas and I am grateful for his group's hospitality at Imperial College, where this work was done.

\section*{$d=2$: an almost-cyclic construction}

\begin{theorem}\label{quasar}
There is a~$2\times2$ complex matrix~$M$ whose first four powers applied to a fiducial vector generate a SIC-POVM in~$d=2$. 
\end{theorem}

\begin{proof}[Proof (by construction)]
Let $\bfv_0 = (1,0)\in\CC^2$.
If we start with $\bfv_0$ as the first vector of a SIC-POVM it follows from~(\ref{SICcond}) that up to appropriate phases, the remaining~3~vectors (in dimension~$d=2$) must be of the form~$(\frac{1}{\sqrt{3}},\sqrt{\frac{2}{3}}e^{i\theta_j})$, for some angles~$\theta_j\in[0,2\pi),\ j=1,2,3$. So if we postulate the existence of a~$2\times2$ matrix~$M$ which begins with~$\bfv_0$ and cycles us around to three more vectors~$\bfv_1=M\bfv_0$, $\bfv_2=M^2\bfv_0$, $\bfv_3=M^3\bfv_0$ then the first column of~$M$ must be of this same form. We do \bf not \rm insist that $M$ be unitary: this would be unnecessarily restrictive given that we are only looking for equiangular \emph{lines}, not necessarily unit vectors. As it turns out the matrix that we end up constructing does in fact generate a sequence of four \emph{unit} vectors, but its eigenvalues are not of modulus one and so subsequent powers give non-unit vectors.

Since any SIC-POVM in dimension~2 may be represented as a tetrahedron of vectors in the Bloch sphere, we may unitarily rotate it so that any chosen pair of its representative vectors lies in the $X,Z$-plane. Hence these two vectors may be viewed as \emph{real} vectors in the sense that their coordinates in the computational basis of $\CC^2$ are real numbers. So we may take the form of $M$ to be:
$$
M = \frac{1}{\sqrt{3}}\left(
\begin{array}{cc}
1&re^{i\rho}\\
\sqrt{2}&se^{i\sigma}\\
\end{array}
\right)
$$
for appropriate non-negative real numbers $r,s,\rho,\sigma$. For any integer~$j$ we shall write $\bfv_j=M^j\bfv_0$. If we write out the equations governing the absolute values of the inner products between the vectors~$\{\bfv_0,\ \bfv_1\}$ and the vector~$\bfv_2$ and try to solve them so that they satisfy equation~(\ref{SICcond}) then we see a neat solution for~$\langle\bfv_0,\bfv_2\rangle$ is $r=\frac{1}{\sqrt{2}}$, $\rho=\frac{\pi}{3}$.
Moving on to~$\langle\bfv_1,\bfv_2\rangle$ then gives us another `obvious' solution as~$s=2$,~$\sigma=\frac{4\pi}{3}$. Somewhat surprisingly, it turns out that this solution for~$\bfv_2$ which was picked only because it was easy to understand, goes on to generate a fourth vector~$\bfv_3$ which has precisely the desired angles with the previous three. So we have a~SIC-POVM
$$\{\bfv_j:j=0,1,2,3\}$$
all generated from the initial vector $\bfv_0$ by successive multiplication by the single matrix
$$
M = \frac{1}{\sqrt{3}}\left(
\begin{array}{cc}
1&\frac{1}{\sqrt{2}}e^{\frac{i\pi}{3}}\\
\sqrt{2}&-2e^{\frac{i\pi}{3}}\\
\end{array}
\right).
$$
For completeness we list the SIC-POVM vectors as
$$
\bfv_0=\begin{pmatrix}1\\0\end{pmatrix},\ \bfv_1=\frac{1}{\sqrt{3}}\begin{pmatrix}1\\\sqrt{2}\end{pmatrix},\ \bfv_2=\frac{i}{\sqrt{3}}\begin{pmatrix}e^{\frac{-i\pi}{3}}\\-\sqrt{2}\end{pmatrix},\ \bfv_3=\frac{1}{\sqrt{3}}\begin{pmatrix}1\\-\sqrt{2}e^{\frac{-i\pi}{3}}\end{pmatrix}.
$$
\end{proof}

So it seems our matrix~$M$ is able to twist~$\bfv_0$ and the next~2 successive vectors~$\bfv_1,\bfv_2$ by exactly the right amount in order to manufacture a SIC-POVM; thereafter (on both sides, ie for positive and negative powers of~$M$) the vectors sacrifice the angle and begin to grow in magnitude. For example~$\bfv_{-1}$ and~$\bfv_4$ each have length~$\sqrt{2}$ and the magnitudes go on to grow symmetrically about the SIC-POVM from there onwards~(see below). It is as though the behaviour is perfectly constrained just while we need it to be, then it shakes off the constraints and spins off to infinity.

The eigenvalues of~$M$ are~$\lambda_{\pm}=-\frac{i}{2}\pm\frac{1}{2}\sqrt{1+2\sqrt{3}i}$, so since they differ in magnitude it follows that the limiting behaviour of~$M^r\bfv_0$ as  $r\rightarrow\pm\infty$ is for the vectors to head towards infinity in magnitude in both directions, with the Hermitian angle between successive vectors $\bfv_j$ and $\bfv_{j+1}$ tending to zero; however with the limiting \emph{pseudo-angle}~\cite[\S2]{scharnhorst}  between successive vectors equal to the argument of the relevant eigenvalue (ie~$\lambda_{-}$ as~$r\rightarrow\infty$ and~$\lambda_{+}$ as~$r\rightarrow-\infty$).

The following image, taken from~\cite{nasa}, may help to visualise the behaviour of this matrix:

\includegraphics[scale=0.58]{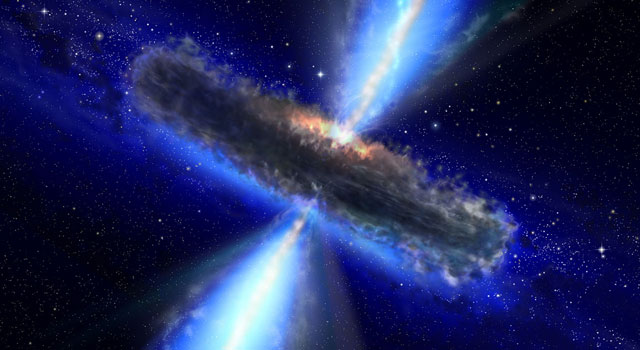}

\noindent where the fiducial vector lies somewhere in the centre, there is a major cluster of vectors of constrained length generated around the centre at the heart of which is the `glowing light' of the particular SIC-POVM configuration, but it gradually (then exponentially) diverges in both directions; the central beam depicts the fact that the powers of~$M$ end up converging to the same vector in a Hermitian angle sense; whereas the widening beam schematically represents the constant non-zero pseudo-angle between successive vectors, which becomes more significant in absolute (Euclidean distance) terms as the vectors grow in magnitude. If we begin at the central point of the series, between~$\bfv_1$ and~$\bfv_2$, then these vectors yield a sequence of integers representing the squared absolute values in both directions as follows:
$$1,\ 1,\ 2,\ 3,\ 5,\ 9,\ 15,\ 26,\ 45,\ 77,\ 133,\ 229,\ 394,\ 679,\ 1169,\ 2013,\ 3467,\ 5970,\ \ldots$$
This sequence does not appear in any of Sloan's online integer sequences~\cite{sloan}.

Another way of visualising the symmetry of this SIC-POVM is to consider what happens if we interpolate the infinite sequence~$\ldots,\ \bfv_0,\ \bfv_1=M\bfv_0,\ \bfv_2=M^2\bfv_0,\ \bfv_3=M^3\bfv_0,\ \ldots$ using any matrix square root of~$M$ (notice the eigenvalues tell us that~$M$ has precisely four (similarity classes of) distinct square roots~\cite[p54]{horn}). Choose any such matrix~$Q$ with~$Q^2=M$. Then the central part (namely the part in which we are most interested) can be indexed instead as
$$\bfv_0=\bfu_{-3/2},\ \bfv_1=\bfu_{-1/2},\ \bfv_2=\bfu_{1/2},\ \bfv_3=\bfu_{3/2},$$
where the subscripts this time refer to half-integral powers of~$M$ as applied to
a central vector~$\bfu_0=Q^3\bfv_0$.

One final curious fact is that the fourth power of~$M$ takes~$\bfv_0$ to the non-unit vector~$\bfv_4=(0,\sqrt{2})$ (which spans the subspace orthogonal to~$\bfv_0$). Let~$\bfu_0=(0,1)$ be a unit vector in the direction of~$\bfv_4$. If we now set~$B=(M^\dagger)^{-1}$ and define~$\bfu_r=B^r\bfu_0$ then the set~$\{\bfu_0,\bfu_1,\bfu_2,\bfu_3\}$ also forms a SIC-POVM which is a kind of `dual' to the above in that for all integers~$j$ by the properties of the inner product,
$$\langle\bfu_j,\bfv_j\rangle=\langle(M^\dagger)^{-j}\bfu_0,M^j\bfv_0\rangle=\langle(M^\dagger)^{-j}(M^\dagger)^j\bfu_0,\bfv_0\rangle=\langle\bfu_0,\bfv_0\rangle=0.$$ 
This is not however the natural dual coming from the adjoint structure - it depends seemingly upon the orthogonality of~$\bfv_0$ and~$\bfv_4$, something which \emph{a priori} is unexpected. If we denote by~$X$ the Pauli~$X$ matrix~$X = \left( \begin{smallmatrix} 0&1\\ 1&0 \end{smallmatrix} \right)$ which is the involution which flips~$\bfv_0$ and~$\bfu_0$, then saying that the~$\bfu_j$ form a SIC-POVM is the same as saying that the matrix~$XM^\dagger X$ also generates a SIC-POVM from the fiducial~$\bfv_0$.



\section*{$d=2,3$: a bi-cyclic structure}

Motivated by the `shape' of the~SIC-POVM constructed in the previous section we began to look for an exact algebraic solution in dimensions~$d=2,3$ starting with a couple of simple assumptions about structure. Such solutions proved relatively straightforward in these low dimensions. In addition for~$d=2$ there is a kind of internal exponential structure to this exact solution, which we shall explain below. However these techniques in their original form cannot be extended to higher dimensions.

\begin{theorem}\label{d2d3}
Let~$d=2$ or~3. There exists a~$d\times d$ unitary matrix~$U_d$ of multiplicative order~$d$ which takes a fiducial vector~$\bfv_0$ to a set of~$d$ vectors~$\bfv_0,\bfv_1,\ldots,\bfv_{d-1}$, each of which represents one of the orbits~$\mathcal{O}_0,\mathcal{O}_1,\ldots,\mathcal{O}_{d-1}$ generated under left multiplication by a fixed~$d\times d$ diagonal unitary matrix~$D_d$ of multiplicative order~$\binom{d+1}{2}$. The disjoint union of these~$d$ orbits is a SIC-POVM.
\end{theorem}

\begin{proof}
Once again the proof is by construction.
For general~$d$ it is a fact of linear algebra~\cite[theorem~2.3.1]{horn} that given any basis of $\CC^d$ we can find unitaries to change the basis to one in which these~$d$ initial column vectors form an upper-triangular matrix. Now given any SIC-POVM set of~$d^2$ vectors it is always possible to take a subset of~$d$ vectors which forms a basis, and therefore in view of the result just stated we may choose these~$d$ such that following an appropriate unitary transformation the column vectors may be arranged to form an upper-triangular matrix. This observation will allow us to construct SIC-POVMs with a particularly transparent geometric structure, because once we have the triangular basis we multiply our basis vectors by a diagonal matrix whose non-zero entries are phases, to create a series of~$d$ orbits, each of which is determined -- by virtue of the `triangular' and diagonal substructures -- solely by the number of non-zero entries in the vector. So our SIC-POVM is then automatically partitioned into~$d$ orbits under the diagonal matrix~$D_d$ and we cycle between the orbits using a unitary matrix~$U_d$ of order~$d$, which we shall construct below.

For any complex vector or matrix~$N$ we denote its transpose by~$N^T$, its entrywise complex conjugate by~$N^*$ and its conjugate transpose by~$N^\dagger={N^*}^T$. Let $\{\bfv_j\}$ be a basis for $\CC^d$ and let $\{\bfw_k\}$ be its dual basis, so $\bfw_k^\dagger\bfv_j=\langle\bfw_k,\bfv_j\rangle=\delta_{kj}$ for all~$j,k$, where~$\delta_{kj}$ is the Kronecker delta. We would like to find a unitary matrix~$U$ which cycles between these vectors, so that for all~$k$:
$$U\bfv_k = \bfv_{k+1}$$
(where we understand the subscript indices as cycling modulo~$d$). I am grateful to Marcus Appleby for pointing out the following lemma, which shows that this is possible if and only if the Gram matrix~$G_\bfv$ of the chosen basis~$\{\bfv_k\}$ is circulant, that is $\langle\bfv_j,\bfv_k\rangle=\langle\bfv_{j+1},\bfv_{k+1}\rangle$ for all $j,k$.

\begin{lemma}\label{unicirc}
With notation as above, let~$\mathcal{A}$ be a~$d\times d$ complex matrix satisfying the following equivalent conditions:

(i) $\mathcal{A}\bfv_j = \bfv_{j+1}$ for all $j$

(ii) $\mathcal{A} = \sum_{k=0}^{d-1}\bfv_k\otimes\bfw_{k-1}^\dagger$

Then $\mathcal{A}$ is unitary if and only if~$G_\bfv$ is circulant.
\end{lemma}

\begin{proof}
We first need to prove the assertion that~(i) and~(ii) are equivalent. That (ii) implies (i) follows from the definitions; the converse is a consequence of the fact that since $\{\bfv_j\}$ is a basis for the space and $\{\bfw_k^\dagger\}$ is a basis for the dual space, the set~$\{\bfv_j\otimes\bfw_{k}^\dagger\}$ is a basis for the matrix operator space in which~$\mathcal{A}$ lives.

So assume that~$\mathcal{A}$ is the matrix defined in~(ii): we must show that being unitary under the standard Hermitian inner product, in the sense that~$\mathcal{A}^\dagger\mathcal{A}=\mathcal{A}\mathcal{A}^\dagger=\mathbf{I}_d$ where~$\mathbf{I}_d$ is the~$d\times d$ identity matrix, is equivalent to the Gram matrix~$G_\bfv$ being circulant. Writing out the change-of-basis equations and using the definition of the dual basis, we see that
\begin{equation}\label{cobgram}
\bfv_l = \sum_{k=0}^{d-1} (G_\bfv^T)_{lk}\bfw_k.
\end{equation}
Since~$\{\bfv_j\}$ is a basis it follows that~$\mathcal{A}$ is unitary if and only if $\mathcal{A}^\dagger\mathcal{A}\bfv_l=\bfv_l$ for all~$l$, which means:
$$\sum_{k=0}^{d-1}\bfw_{k-1}\otimes\bfv_k^\dagger\sum_{j=0}^{d-1}\bfv_j\otimes\bfw_{j-1}^\dagger\bfv_l=\bfv_l\hbox{\rm\ for\ all\ }l.$$
Now~$\bfw_{j-1}^\dagger\bfv_l=\delta_{j-1,l}=\delta_{j,l+1}$ and so the terms in the inner sum are non-zero only when~$j=l+1$. So by the definition of the Gram matrix~$G_\bfv$ the sum becomes:
$$\sum_{k=0}^{d-1}(G_\bfv)_{k,l+1}\bfw_{k-1}=\bfv_l\hbox{\rm\ for\ all\ }l.$$
Using the index~$k+1$ in place of~$k$ and transposing gives
$$\sum_{k=0}^{d-1}(G_\bfv^T)_{l+1,k+1}\bfw_k=\bfv_l\hbox{\rm\ for\ all\ }l$$
and since~$\{\bfv_j\}$ and~$\{\bfw_k\}$ are bases,~(\ref{cobgram}) shows that each of the above statements is equivalent to
$$(G_\bfv)_{m,n}=(G_\bfv)_{m+1,n+1}$$
for all $m,n$. This completes the proof of the lemma.
\end{proof}

Let us specialise to the case~$d=2$ or~$3$ with our initial vector $\bfv_0$ which is $(1,0)$ for $d=2$ and $(1,0,0)$ for $d=3$. Armed with the above lemma we now search for~$\bfv_1,\ldots,\bfv_{d-1}$ such that the basis $\{\bfv_k\}$ has upper-triangular form and such that the Gram matrix~$G_\bfv$ is circulant. Since it is also automatically Hermitian this reduces considerably the possibilities for the vectors. Henceforth all of the vectors we consider will be assumed to be unit vectors.

\ 

$\mathbf{d=2:}$ If we perform the same trick as in the previous section by identifying any SIC-POVM in dimension~2 with a tetrahedron in the Bloch sphere then we may assume once again that our second vector~$\bfv_1$ is~$\frac{1}{\sqrt{3}}(1,\sqrt{2})$. Notice that this automatically fulfils the circulant criterion, since in dimension~2 it boils down to the single requirement that~$\langle\bfv_0,\bfv_1\rangle=\langle\bfv_1,\bfv_0\rangle$ which by the fact that the inner product is Hermitian forces both to be real. So we may write our candidate for a unitary matrix which cycles between~$\bfv_0$ and~$\bfv_1$ as
$$U_2=\frac{1}{\sqrt{3}}\left(
\begin{array}{cc}
1&\alpha\\
\sqrt{2}&\beta\\
\end{array}
\right)
$$
for some complex numbers~$\alpha,\beta$. If we require that~$U_2$ be unitary and of multiplicative order~2 it follows that in fact~$U_2$ must be Hermitian and so~$\alpha=\sqrt{2}$ and~$\beta=\pm1$. Writing out the equations for~$U_2^2$ we find that the only possibility is:
$$U_2=\frac{1}{\sqrt{3}}\left(
\begin{array}{cc}
1&\sqrt{2}\\
\sqrt{2}&-1\\
\end{array}
\right),
$$
and we may verify that indeed~$U_2\bfv_0=\bfv_1$ and~$U_2\bfv_1=\bfv_0$.
We now look for a diagonal matrix~$D_2$ of phases which will take our initial vectors~$\bfv_0$ and~$\bfv_1$ by left multiplication to~2 more vectors which comprise the remaining part of the generators for a maximal set of equiangular lines. Notice that the upper left-hand entry of~$D_2$ must be~1, since~$\bfv_0$ is always in a~$D_d$-orbit of its own (all other vectors in any SIC-POVM containing~$\bfv_0$ are forced to have their first entry equal to a phase times~$\frac{1}{\sqrt{(d+1)}}$). So our diagonal matrix in this case will look like
$$D_2 = \left(
\begin{array}{cc}
1&0\\
0&\zeta\\
\end{array}
\right)
$$
where~$\zeta$ is some phase. We set~$\bfv_2=D_2\bfv_1$ and~$\bfv_3=D_2\bfv_2=D_2^2\bfv_1$. Writing out the equations for the set~$\{\bfv_0,\bfv_1,\bfv_2,\bfv_3\}$ to form a spanning set for a maximal set of equiangular lines in dimension~2 we observe first that the equiangularity between~$\bfv_0$ and the other three is automatic, by our choice of first entries (see the discussion of~$d=3$ below for a deeper insight into this property, which is the essence of the advantage of this construction method). So we only need worry about the angles among the remaining vectors~$\bfv_1,\bfv_2$ and~$\bfv_3$, which boil down to just three equations of the form
$$\vert1+2\zeta^r\vert=\sqrt{3},$$
where~$r=1$ or 2. This forces~$\zeta$ to be one of the primitive cube roots of unity, and we are done. Notice that the requirement that $D_2^3=\mathbf{I}_2$ would also have forced~$\zeta$ to be one of the cube roots of unity (without necessarily having been a solution which provided a SIC-POVM!). However we did not impose this \emph{a priori} in case a similar situation should arise to that in the first section, where the generating matrix was not of finite order.

\begin{remark}
The way in which the above example and its counterpart below in dimension~3 were originally discovered was by considering `Hadamard' multiplication of rank~1 projectors with the density matrices corresponding to the upper-triangular vector set, since the structure shone through much more clearly there than in any other format; presumably because the phase ambiguities are removed. If we consider that our matrix~$D_2$ is in fact the diagonal matrix of a vector~$\bfh=(1,-e^\frac{\pi i}{3})$ say, and if we form the rank~1 projector from~$\bfh$ which is the Hermitian matrix~$H = \bfh^\dagger\bfh = \left( \begin{smallmatrix} 1&-e^\frac{-\pi i}{3}\\ -e^\frac{\pi i}{3}&1 \end{smallmatrix} \right)$, then the following remarkable fact arises: the set
$$\big\{\bfv_0,\ e^{i\theta_m H}\bfv_0,\ (H\ast e^{i\theta_m H})\bfv_0,\ (H\ast H\ast e^{i\theta_m H})\bfv_0\big\}$$ 
or equivalently
$$\big\{\bfv_0,\ e^{i\theta_m H}\bfv_0,\ (e^{i\theta_m H}\bfv_0)\ast\bfh,\ (e^{i\theta_m H}\bfv_0)\ast\bfh\ast\bfh\big\}$$
is a SIC-POVM, where~$\theta_m$ denotes the so-called~\emph{magic angle}~$\theta_m=\arccos{\frac{1}{\sqrt{3}}}$, and where the~$\ast$ denotes Hadamard (elementwise) multiplication of vectors and/or matrices. What we lose however in this version is the finite order property of the transition unitary~$e^{i\theta_m H}$: while (under Hadamard multiplication) the matrix~$H$ still has finite order, the unitary matrix~$e^{i\theta_m H}$ has infinite multiplicative order. In this context we mention that our original transition matrix~$U_2$ above may be expanded as the exponential
$$U_2 = e^{-i\theta_m Y}Z,$$
where~$Y=\left( \begin{smallmatrix} 0&-i\\ i&0 \end{smallmatrix} \right)$ and~$Z=\left( \begin{smallmatrix} 1&0\\ 0&-1 \end{smallmatrix} \right)$ are the usual Pauli matrices.
\end{remark}

\ 

$\mathbf{d=3:}$  This time our vector~$\bfv_0=(1,0,0)$ and we must find a unitary matrix~$U_3$ which takes us from~$\bfv_0$ cyclically to vectors~$\bfv_1$ and~$\bfv_2$ which have respectively~2 and~3 non-zero entries (the upper-triangular format referred to above). We know by the same argument as in dimension~2 that the first entry of each of these vectors must have absolute value~$\frac{1}{2}$, so let the top entry of~$\bfv_1$ be~$\frac{1}{2}e^{ix}$ for some~$x\in[0,2\pi)$. The hypothesis that the Gram matrix of the set~$\{\bfv_0,\bfv_1,\bfv_2\}$ be circulant in particular forces~$\frac{1}{2}e^{ix}=\langle\bfv_0,\bfv_1\rangle=\langle\bfv_2,\bfv_0\rangle$ and so the top entry of~$\bfv_2$ must equal~$\frac{1}{2}e^{-ix}$. 
So let us write
$$\bfv_1=\begin{pmatrix}\frac{1}{2}e^{ix}\\\frac{\sqrt{3}}{2}e^{iy}\\0\end{pmatrix},\ \bfv_2=\begin{pmatrix}\frac{1}{2}e^{-ix}\\re^{i\eta}\\\sqrt{\frac{3}{4}-r^2}e^{i\kappa}\end{pmatrix}$$
for suitable real non-negative~$y,r,\eta,\kappa$. We remark first that~$\kappa$ may be set to be zero since it has no impact upon any other quantities, including the effect of our target~$D_3$ matrix, as we shall explain below. It remains to ensure that the middle inner product~$\langle\bfv_1,\bfv_2\rangle$ then also equals~$\frac{1}{2}e^{ix}$. (Notice that the other~3 non-diagonal inner products in the Gram matrix are forced to obey the same circulant rule here because the Gram matrix is Hermitian and the dimension is only~3). 
So we only need solve the equation:
\begin{equation}\label{ranga}
\frac{1}{2}e^{ix}=\langle\bfv_1,\bfv_2\rangle=\frac{1}{4}e^{-2ix}+\frac{\sqrt{3}}{2}re^{i(\eta-y)},
\end{equation}
which upon multiplying throughout by~$2e^{2ix}\neq0$ becomes
$$(e^{ix})^3 - \sqrt{3}re^{i(\eta-y)}(e^{ix})^2 - \frac{1}{2} = 0.$$
Viewed as an equation in the variable~$e^{ix}$ and bearing in mind the role of sixth roots of unity in this theory, this equation has a particularly suggestive form: namely if we take the phase~$e^{i(\eta-y)}(e^{ix})^2$ in the central term to be~$\pm i$ then the whole equation has the shape of a sixth root of unity minus its real and imaginary components. That is, if we set~$(e^{ix})^3 = e^{\frac{\pi i}{3}} = \frac{1}{2}+\frac{\sqrt{3}}{2}i$, set~$r=\frac{1}{2}$ and ensure that the phase in the middle term is equal to~$i$, then we have a solution. So one neat form is to set~$\eta=\frac{\pi}{2}$,~$y=\frac{2\pi}{9}$ and so the vectors become:
$$\bfv_0=\begin{pmatrix}1\\0\\0\end{pmatrix},\ \bfv_1=\begin{pmatrix}\frac{1}{2}e^{\frac{\pi i}{9}}\\\frac{\sqrt{3}}{2}e^{\frac{2\pi i}{9}}\\0\end{pmatrix},\ \bfv_2=\begin{pmatrix}\frac{1}{2}e^{-\frac{\pi i}{9}}\\\frac{i}{2}\\\frac{1}{\sqrt{2}}\end{pmatrix}$$
and using the formula in~(ii) of lemma~\ref{unicirc} gives our transition unitary~$U_3$ to be:
$$U_3=\left(
\begin{array}{ccc}
\frac{1}{2}e^{\frac{\pi i}{9}}                  &    -\frac{i}{2}              &        \frac{1}{\sqrt{2}}                    \\
\frac{\sqrt{3}}{2}e^{\frac{2\pi i}{9}}        &    \frac{i}{2\sqrt{3}}  e^{\frac{\pi i}{9}}                       &       - \frac{1}{\sqrt{6}}  e^{\frac{\pi i}{9}}                  \\
0                               &     \sqrt{\frac{2}{3}}e^{\frac{-2\pi i}{9}}              &              - \frac{i}{\sqrt{3}}   e^{\frac{-2\pi i}{9}}                   \\
\end{array}
\right)
$$
which has multiplicative order~3. So we have our substructure of a triangular basis.

It remains to search for a diagonal matrix~$D_3$ of phases such that the (subspaces generated by the) orbits of these vectors under left multiplication by~$D_3$ do in fact constitute a full set of equiangular lines. As in the~$d=2$ case the top left-hand entry of~$D_3$ must be~1. So let us write
$$D_3 = \left(
\begin{array}{ccc}
1&0&0\\
0&\xi&0\\
0&0&\zeta\\
\end{array}
\right)
$$
for some phases~$\xi,\zeta$. We observe that $D_3\bfv_1=\begin{pmatrix}\frac{1}{2}e^{\frac{\pi i}{9}}\\\frac{\sqrt{3}}{2}e^{\frac{2\pi i}{9}}\xi\\0\end{pmatrix}$ and so~$\langle\bfv_1,D_3\bfv_1\rangle=\frac{1}{4}+\frac{3}{4}\xi$. For this to be of absolute value~$\frac{1}{2}$ we require that~$\xi=-1$. Substituting this in turn into the equation for~$\langle D_3\bfv_1,\bfv_2\rangle$ yields an inner product~$\frac{1}{2}e^{-\frac{5\pi i}{9}}$, which is also of the correct absolute value. So far, so good: we have a collection of four vectors which span four equiangular lines. The final step is to check whether there is an appropriate choice of~$\zeta$ to generate the other five. 

Returning for a moment to the case of general~$d$, observe that for any positive integers~$r,s$ and~$j$, since~$D_d$ is by construction unitary:
$$\langle D_d^{r}\bfv_{s+j},\bfv_{s}\rangle = \langle\bfv_{s+j},{D_d^\dagger}^r\bfv_{s}\rangle = \langle\bfv_{s+j},D_d^{n-r}\bfv_{s}\rangle,$$
where~$n$ is the lowest common multiple of the orders of the eigenvalues chosen so far for~$D_d$. In other words, all of the vectors in orbit~$\mathcal{O}_{s+j}$ will have the correct Hermitian angle with all of those in orbit~$\mathcal{O}_{s}$, since by stage~$(s+j)$ we have already verified that~$\bfv_{s+j}$ makes the correct angle with all of orbit~$\mathcal{O}_{s}$ and since~$D_d$ does not affect anything in the vectors of orbit~$\mathcal{O}_{s}$ beyond the $s$-th entry, the same must be true of all of the~$D_d$-multiples of~$\bfv_{s+j}$ no matter what our choice of eigenvalue at the~$(s+j)$-level. 
So the point about the upper-triangular structure we have created may be seen here (for~$d=2$ it was rather trivial): once we have created~$k$ levels in the sense that we have vectors~$\bfv_0,\ldots,\bfv_{k-1}$ and all of their finite orbits~$\mathcal{O}_0,\ldots,\mathcal{O}_{k-1}$ under repeated multiplication by~$D_d$, and once we are sure that the subsequent vectors~$\bfv_{k},\ldots,\bfv_{d-1}$ make the correct Hermitian angle with all of these orbits, then we may choose~\emph{any} phases for the~${k},\ldots,(d-1)$-st eigenvalues of~$D_d$ safe in the knowledge that the images of the vectors~$\bfv_{k},\ldots,\bfv_{d-1}$ under any power of the resulting matrix~$D_d$ will automatically make the correct Hermitian angle with the orbits~$\mathcal{O}_0,\ldots,\mathcal{O}_{k-1}$. So we are reduced at each $k+1$-st stage to ensuring that the set of new vectors~$\{D_d^r\bfv_{k}\}$ has the correct set of mutual angles with one another and with the subsequent vectors~$\bfv_{k+1},\ldots,\bfv_{d-1}$; the previous orbits automatically `fall into line'. This also shows that within each level we only need to check~$\vert\mathcal{O}_k\vert$ equations rather than the usual~$\binom{\vert\mathcal{O}_k\vert}{2}$, since for any integers~$r,s$:
$$\langle D_d^{r}\bfv_k,D_d^s\bfv_k\rangle = \langle\bfv_k,{D_d^\dagger}^rD_d^s\bfv_k\rangle = \langle\bfv_k,D_d^{s-r}\bfv_k\rangle.$$


So in dimension~3 it is a consequence of the above discussion that no matter what our choice of~$\zeta$, the vectors~$D_3^t\bfv_2$ for integer~$t$ will always have the correct angle with vectors~$\bfv_0$,~$\bfv_1$, and~$D_3\bfv_1$. So we only need to focus on the inner products between the vectors~$D_3^t\bfv_2$ for~$t=0,1,2,3,4,5$. A glance at the shape of the vector~$\bfv_2$ shows that for any integers~$s,t$, since $D_3$ is automatically unitary:
$$\langle D_3^s\bfv_2,D_3^t\bfv_2\rangle = \langle\bfv_2,D_3^{t-s}\bfv_2\rangle = \frac{1}{4} + (-1)^{t-s}\frac{1}{4}  + \frac{1}{2}\zeta^{t-s},$$
explicitly showing that the individual vectors are unit vectors when~$s=t$. Without loss of generality we may assume when~$s\neq t$ that~$0\leq s<t\leq5$, so in particular~$1\leq t-s\leq 5$.
The above expression shows immediately that if $t-s$ is odd then we have the correct absolute value of~$\frac{1}{2}$; when~$t-s$ is even (ie equal to~2 or~4) one sees that any primitive cube root or indeed sixth root of unity will once again yield the correct absolute value of~$\frac{1}{2}$. So for simplicity we shall set
$$\zeta = e^\frac{2\pi i}{3},$$
hence~$D_3$ has the form
$$D_3 = \left(
\begin{array}{ccc}
1&0&0\\
0&-1&0\\
0&0&e^\frac{2\pi i}{3}\\
\end{array}
\right),
$$
whence our full set of vectors is:
\begin{eqnarray*}
\mathcal{O}_0 & = & \{ \begin{pmatrix}1\\0\\0\end{pmatrix} \}, \\
\mathcal{O}_1 & = & \{ \begin{pmatrix}\frac{1}{2}e^{\frac{\pi i}{9}}\\\frac{\sqrt{3}}{2}e^{\frac{2\pi i}{9}}\\0\end{pmatrix}, \begin{pmatrix}\frac{1}{2}e^{\frac{\pi i}{9}}\\-\frac{\sqrt{3}}{2}e^{\frac{2\pi i}{9}}\\0\end{pmatrix} \},  \\
\mathcal{O}_2 & = & \{ \begin{pmatrix}\frac{1}{2}e^{\frac{-\pi i}{9}}\\\frac{i}{2}\\\frac{1}{\sqrt{2}}\end{pmatrix}, \begin{pmatrix}\frac{1}{2}e^{\frac{-\pi i}{9}}\\-\frac{i}{2}\\\frac{1}{\sqrt{2}}e^\frac{2\pi i}{3}\end{pmatrix}, \begin{pmatrix}\frac{1}{2}e^{\frac{-\pi i}{9}}\\\frac{i}{2}\\\frac{1}{\sqrt{2}}e^\frac{4\pi i}{3}\end{pmatrix}, \begin{pmatrix}\frac{1}{2}e^{\frac{-\pi i}{9}}\\-\frac{i}{2}\\\frac{1}{\sqrt{2}}\end{pmatrix}, \begin{pmatrix}\frac{1}{2}e^{\frac{-\pi i}{9}}\\\frac{i}{2}\\\frac{1}{\sqrt{2}}e^\frac{2\pi i}{3}\end{pmatrix}, \begin{pmatrix}\frac{1}{2}e^{\frac{-\pi i}{9}}\\-\frac{i}{2}\\\frac{1}{\sqrt{2}}e^\frac{4\pi i}{3}\end{pmatrix} \}. \\
\end{eqnarray*}
Notice we have split it into its three natural $D_3$-orbits:~$\mathcal{O}_0$ generated by~$\bfv_0$,~$\mathcal{O}_1$ generated by~$\bfv_1$ and~$\mathcal{O}_2$ generated by~$\bfv_2$. Also we remark that~$D_3^6=\mathbf{I}_3$, so in fact in this case we are able to stick to finite-order unitaries both for the transition matrix between orbits, and for the diagonal matrix which generates each orbit.

This completes the proof of theorem~\ref{d2d3}.\end{proof}

We should mention that the mere creation of the initial set~$\{\bfv_0,\bfv_1,\bfv_2\}$ does not in any way guarantee that it can be extended to a SIC-POVM in the above fashion. For example it is possible to create a set of three totally real vectors (using~$x=0$ above and then solving equation~(\ref{ranga})) which have no corresponding diagonal matrix to extend them to a full set.

\begin{remark}
Any attempt to extend this methodology beyond~$d=3$ using the na\"ive diagonal approach which worked in~$d=2,3$ is unfortunately doomed to fail - in a sense one `runs out of degrees of freedom' far too quickly. This does not rule out a kind of `block diagonal' approach, which we hope to be the subject of future work.
\end{remark}

\end{document}